%% file: mobicom18.tex
\documentclass[10pt,sigconf]{acmart}
\renewcommand\footnotetextcopyrightpermission[1]{} 
\setcopyright{acmcopyright}

\usepackage[utf8]{inputenc}
\usepackage[T1]{fontenc}
\usepackage{textcomp}
\usepackage{gensymb}
\usepackage{graphicx}
\usepackage{url}
\usepackage{amsmath}
\usepackage{latexsym,amssymb,epsfig,color,verbatim}
\usepackage{algorithm}
\usepackage{algpseudocode}
\usepackage{pifont}
\usepackage{slashbox}
\usepackage[tight,footnotesize,center]{subfigure}
\usepackage{color}
\usepackage{cancel}
\usepackage{multirow}
\usepackage{amsthm}
\usepackage{bm}
\usepackage{tabu}
\usepackage{mathtools}
\usepackage{algorithm}
\usepackage{algpseudocode}
\usepackage{mathtools,amssymb,lipsum}

\newcommand{\BEQA}{\begin{eqnarray}}
\newcommand{\EEQA}{\end{eqnarray}}

\newtheorem{theorem}{{\bf Theorem}}

\theoremstyle{definition}
\newtheorem{definition}{Definition}[section]

\makeatletter
\def\@copyrightspace{\relax}
\makeatother
\setcopyright{none} 
\begin{document}
\title{A Game-Theoretic Approach to Multi-Objective Resource Sharing and Allocation in Mobile Edge Clouds}

\author{Faheem Zafari}
\affiliation{%
  \institution{Imperial College London}
}

\author{Jian Li}
\affiliation{%
  \institution{UMass Amherst}
 }

\author{Kin K. Leung}
\affiliation{%
  \institution{Imperial College London}
}

\author{Don Towsley}
\affiliation{%
  \institution{UMass Amherst}
}
\author{Ananthram Swami}
\affiliation{%
 \institution{U.S. Army Research Laboratory}
}

\begin{abstract}
Mobile edge computing seeks to provide resources to different delay-sensitive applications. However, allocating the limited edge resources to a number of applications is a challenging problem. To alleviate the resource scarcity problem, we propose sharing of resources among multiple edge computing service providers where each service provider has a particular utility to optimize.  We model the resource allocation and sharing problem as a multi-objective optimization problem and present a \emph{Cooperative Game Theory} (CGT) based framework, where each edge service provider first satisfies its native applications and then shares its remaining resources (if available) with users of other providers. Furthermore, we propose an $\mathcal{O}(N)$ algorithm that provides  allocation decisions from the \emph{core}, hence the obtained allocations are \emph{Pareto} optimal and the grand coalition of all the service providers is stable.  Experimental results show that our proposed resource allocation and sharing framework improves the utility of all the service providers compared with the case where the service providers are working alone (no resource sharing).  Our $\mathcal{O}(N)$ algorithm reduces the time complexity of obtaining a solution from the core by as much as 71.67\% when compared with the \emph{Shapley value}.
\end{abstract}

 \copyrightyear{2018} 
\acmYear{2018} 
\setcopyright{usgovmixed}
\acmConference[EdgeTech'18: ]{2018 Technologies for the Wireless Edge Workshop}{November 2, 2018}{New Delhi, India}
\acmBooktitle{2018 Technologies for the Wireless Edge Workshop (EdgeTech'18: ), November 2, 2018, New Delhi, India}
\acmPrice{15.00}
\acmDOI{10.1145/3266276.3266277}
\acmISBN{978-1-4503-5931-3/18/11}

%
%

\maketitle
\section{Introduction}\label{sec:intro}
\input{01-intro}

\section{Preliminaries}\label{sec:prelim}
\input{preli}
\section{System Model}\label{sec:sysmodel}
\input{02-sysmodel}

\subsection{Optimization Problem}\label{sec:opt_problem}
\input{03-opt-problem}

\subsection{Game Theoretic Solution}\label{sec:gametheory}
\input{game-theory}
\section{Experimental Results}\label{sec:exp_results}
\input{04-exp-results}
\section{Conclusions}\label{sec:conclusion}
\input{05-conclusions}

\bibliographystyle{unsrt}
\bibliography{refs}

\end{document}

%% file: 01-intro.tex
\footnote{This work was supported by the U.S. Army Research Laboratory and the U.K. Ministry of Defence under Agreement Number W911NF-16-3-0001. The views and conclusions contained in this document are those of the authors and should not be interpreted as representing the official policies, either expressed or implied, of the U.S. Army Research Laboratory, the U.S. Government, the U.K. Ministry of Defence or the U.K. Government. The U.S. and U.K. Governments are authorized to reproduce and distribute reprints for Government purposes notwithstanding any copy-right notation hereon.  Faheem Zafari also acknowledges the financial support by EPSRC Centre for Doctoral Training in High Performance Embedded and Distributed Systems (HiPEDS, Grant Reference EP/L016796/1), and Department of Electrical and Electronics Engineering, Imperial College London.}Mobile edge computing is a viable solution to support resource intensive applications (users).  Edge computing relies on edge clouds placed at the edge of any network \cite{he2018s}. This, in contrast with running applications on different mobile devices or deep in the Internet, usually allows   one-hop communication \cite{jia2015optimal} between   edge clouds and application that results in   reducing application latency. However, a fundamental limitation of mobile edge computing is that in contrast with traditional cloud platforms and data centers,  edge clouds are limited in resources and may not always be able to satisfy  application demands \cite{he2018s}. Realizing the resource scarcity problem, the research community has started several initiatives  to create an open edge computing platform  where  edge clouds in the same geographical location can form a shared resource pool. However, allocating these resources efficiently from the shared pool to different applications in itself is a challenge. 
\par There have been several  attempts in the literature  to address the resource allocation problem. He et al. \cite{he2018s} studied the allocation of edge resources to different applications by jointly considering request scheduling and service placement.  
Jia et al. \cite{jia2015optimal} discussed edge cloud placement and allocation of  resources to mobile users in the edge cloud in a Wireless Metropolitan Area Network (WMAN). 
Xu et al. \cite{xu2016efficient} discussed edge cloud placement in a large-scale WMAN that contains multiple wireless access points (APs). 
However, most  work does not  account for the fact that edge resources can belong to different service providers where each service provider can have a particular objective to optimize such as security, throughput, latency, etc. Therefore creating a resource pool and then allocating these resources requires taking different service provider objectives into account, which results in a multi-objective optimization \cite{cho2017survey} (MOO) problem.  Furthermore, each service provider has primary applications that should be prioritized over  applications of other service providers as customer loyalty is an important part of the cloud business model. For example, \emph{Amazon Web Services} (AWS) and \emph{Microsoft Azure} need to satisfy the demands of their own customers first before they ``rent''  resources to other service providers. 



\par 

In this paper, we attempt to address the aforementioned shortcomings. We consider an edge computing setting where different edge clouds belonging to different service providers are placed at the network edge. Each cloud has a specific amount of resources and particular applications  affiliated with the cloud  can ask for resources. All clouds initially attempt to allocate resources to their own affiliated applications. If a cloud can satisfy its own applications and still have available resources, it can share them with other edge clouds that might need resources. 
To capture this, we present a \emph{Cooperative Game Theory} (CGT) based resource sharing and allocation mechanism in which different edge clouds share their resources and form a coalition to satisfy the requests of different applications. Our CGT based framework takes into account the fact that different edge clouds  may have different objectives, which is why traditional single objective optimization framework cannot be used. 
The  contributions of this paper are:
\begin{enumerate}
	\item We propose a CGT based multi-objective resource sharing and allocation framework  for edge clouds in an edge computing setting. We show that  the resource sharing and allocation problem can be modeled as a \emph{Canonical} game. The \emph{core} of this canonical game is non-empty and the \emph{Shapley value} \cite{shapley1988value} lies inside the core. For small number of players, we show numerically that the Shapley value  provides a fair, Pareto-optimal and stable\footnote{No service provider has the incentive to leave the coalition and form a smaller coalition.} allocation.  
	\item We  propose an efficient $\mathcal{O}(N)$ algorithm that provides  an allocation from the core, hence reducing the complexity from $\mathcal{O}(2^N)$ (for Shapley value). 
	\item We evaluate the performance of our proposed framework and show that the resource sharing and allocation mechanism  improves the utilities of  game players. 
\end{enumerate}
The paper is  structured as: 
Section~\ref{sec:prelim} provides a primer on  CGT,  the Core and Shapley Value. 
Section~\ref{sec:sysmodel} presents our system model. It also presents the resource sharing and allocation optimization problem with the game theoretic solution.  
Section~\ref{sec:exp_results} presents our experimental results while Section~\ref{sec:conclusion} concludes the paper.


%% file: preli.tex
In this section, we present a primer on  Cooperative Game Theory, the Core and Shapley Value as they are used later in the paper.
\subsection{Cooperative Game Theory}
Cooperative game theory assists us in understanding the behavior of rational players in a cooperative setting \cite{han2012game}. Players can have agreements among themselves that affect the strategies as well as their obtained payoffs or utilities.  
Below we provide some basic definitions related to cooperative game theory. 

\noindent\textit{\bf {Coalition Game\cite{han2012game}:}} Any coalition game can be represented by the pair $(\mathcal{N},v)$ where $\mathcal{N}$ is the set of players that play the game, while $v$ is the mapping function that determines the utilities or payoffs received by the players in $\mathcal{N}$. 


\noindent\textit{\bf {Transferable Utility (TU):}} If the total utility of any coalition can be divided in any manner among the game players, then the game has a \emph{transferable} utility. 

\noindent\textit{\bf {Characteristic function:}} The characteristic function for a coalitional game with TU is a mapping $v: 2^{\mathcal{N}} \mapsto \mathbb{R}$ with $v(\emptyset)=0$.

\noindent\textit{\bf {Superadditivity of TU games:}} Any game with TU is said to be superadditive if the formation of large coalitions is always desired. Mathematically, 
\begin{equation}\label{eq:superadditivity}
v(S_1 \cup S_2 )\geq v(S_1)+v(S_2) \; \forall S_1,  S_2 \subset \mathcal{N}, s.t.\; S_1\cap S_2 =\emptyset
\end{equation}
\noindent\textit{\bf {Canonical Game:}} A coalition game is canonical if it is in the characteristic form and is superadditive. 

\noindent\textit{\bf {Group Rational:}}	A payoff vector $\textbf{x}\in \mathbb{R}^\mathcal{N}$ for dividing $v(\mathcal{N})$ is group-rational if $\sum_{n \in \mathcal{N}}x_n=v(\mathcal{N})$.

\noindent\textit{\bf{ Individual Rational:}}	A payoff vector $\textbf{x}\in \mathbb{R}^\mathcal{N}$ is individually-rational if every player obtains a larger benefit than it would acting alone, i.e., $x_n \geq v(\{n\}), \forall n\in \mathcal{N}$.

\noindent\textit{\bf {Imputation:}}	A payoff vector that is both individual and group rational is known as an imputation. 


\noindent\textit{\bf {Core:}}	For any TU canonical game $(\mathcal{N},v)$, the core is the set of imputations in which no coalition $S\subset\mathcal{N}$ has any incentive to reject the proposed payoff allocation and deviate from the \emph{grand coalition}\footnote{Grand coalition means that all the players in the game form a coalition.} to form a coalition $S$ instead. 
\par 	Any payoff allocation from the core is Pareto-optimal as evident from the definition of the core. Furthermore, the grand coalition formed is stable. 
However, a core is not always guaranteed to exist. Even if a core exists, it might be too large so finding a suitable allocation from the core may not be easy. Furthermore, as seen from the definition, the allocation from the core may not always be fair to all the players.  
 \emph{Shapley value} can be used to address the aforementioned  shortcomings of the core. Details about Shapley value can be found in  \cite{han2012game}. 
 

%% file: 02-sysmodel.tex


Let $\mathcal{N}=\{1,2, \cdots, N\}$ be the set of all the edge clouds that act as players in our game. We assume that each player has a set of $\mathcal{K}=\{1,2,\cdots, K\}$ different types of resources such as communication, computation and storage resources.  The $n$-th edge cloud can report its available resources $C^{(n)} = \{ C_1^{(n)} .... C_K^{(n)}\}$ to a central entity, the coalition controller. Here $C_k^{(n)}$ is the amount of resources of type $k$ available to edge node $n$. 
Vector $C=\{\sum_{n \in \mathcal{N}}C_{1}^{(n)}, \sum_{n \in \mathcal{N}}C_{2}^{(n)}, \cdots,\sum_{n \in \mathcal{N}}C_{K}^{(n)}\}$ represents  all available resources at different edge clouds.  Each edge cloud $n$ has a set of  native applications $\mathcal{M}_n= \{1,2,\cdots, M_n\}$ that ask for resources. Furthermore, the set of all applications that ask for resources from the set of edge clouds (coalition of edge clouds) is given by $\mathcal{M}=\mathcal{M}_1\cup\mathcal{M}_2\cdots\cup \mathcal{M}_N , \;$ where $\mathcal{M}_i \cap \mathcal{M}_j=\emptyset, \; \forall i \neq j,$ i.e., each application asks only  one edge cloud for resources. 
The coalition controller receives a request (requirement) matrix $R^{(n)}$ from every player $n \in \mathcal{N}$ 
	\begin{equation}
\label{eq:R_Req}
R^{(n)}=\Biggl[\begin{smallmatrix}
\mathbf{r^{(\textit{n})}_1}\\ 
.\\ 
.\\ 
.\\
\mathbf{r^{(\textit{n})}_{{M}_n}}
\end{smallmatrix}\Biggr] = \Biggl[\begin{smallmatrix}
r^{(n)}_{11} & \cdots  &\cdots   & r^{(n)}_{1{K}} \\ 
. & . &.  &. \\ 
. & . &.  &. \\
. & . &.  &. \\
r^{(n)}_{M_{n}1}&\cdots  &\cdots   & r^{(n)}_{{M}_n{K}}
\end{smallmatrix}\Biggr]
\end{equation}
where the $i^{th}$ row corresponds to the $i^{th}$ application while  columns represent different resources, i.e., $r_{ij}$ is the amount of $j^{th}$ resource that application $i \in \mathcal{M}_n$ requests.
The coalition, based on $R$ and  $C$, has to make an allocation decision $\mathcal{X}$ that optimizes the utilities $u_n(\mathcal{X})$ of all the edge clouds  $n \in \mathcal{N}$.  The allocation decision $\mathcal{X}$ is  a vector $\mathcal{X}=\{X^{(1)},X^{(2)},\cdots,$\\$ X^{(N)}\}$ that indicates how much of each resource $k \in \mathcal{K}$ is allocated to application $i$ at edge cloud $n \in \mathcal{N}$. Mathematically, 
\begin{equation}
\label{eq:A_n}
X^{(n)}=\Biggl[\begin{smallmatrix}
\mathbf{x_1^{(\textit{n})}}\\ 
.\\ 
.\\ 
.\\
\mathbf{x_{{|\mathcal{M}|}}^{(\textit{n})}}
\end{smallmatrix}\Biggr] = \Biggl[\begin{smallmatrix}
x_{11}^{(n)} & \cdots  &\cdots   & x_{1{K}}^{(n)} \\ 
. & . &.  &. \\ 
. & . &.  &. \\
. & . &.  &. \\
x_{{|\mathcal{M}|}1}^{(n)}&\cdots  &\cdots   & x_{{|\mathcal{M}|}{K}}^{(n)}
\end{smallmatrix}\Biggr]
\end{equation}
where $x_{ik}^{(n)}$ is the amount of resource $k\in \mathcal{K}$ belonging to player $n$ that is allocated  to application $i$. 

%% file: 03-opt-problem.tex
In this section, we first present the resource allocation problem for a single edge cloud (no resource sharing with other edge clouds). Then we present the MOO problem for the coalition. For a single edge cloud $n \in \mathcal{N}$, the allocation decision matrix $X_{SO}^{(n)}$, where $SO$ stands for single objective optimization, is given by 
\begin{equation}
\label{eq:X_n}
X_{SO}^{(\textit{n})}=\Bigl[\begin{smallmatrix}
\mathbf{x_1^{(\textit{n})}} . & . & .& \mathbf{x_{{{M}_n}}^{(\textit{n})}}
\end{smallmatrix}\Bigr]^t
\end{equation}
The optimization problem is given below:
\begin{subequations}\label{eq:opt_single}
	\begin{align}
	\max_{{X}_{SO}^{(n)}}\quad& u_n({X}_{SO}^{(n)}) \label{eq:objsingle}\\
	\text{s.t.}\quad & \sum_{i \in \mathcal{M}_n} x_{ik}^{(n)}\leq C_{k}^{(n)}, \quad \forall k \in \mathcal{K},  
	 \label{eq:singlefirst} \displaybreak[0]\\
	& x_{ik}^{(n)} \leq r^{(n)}_{ik}, \quad \forall\; i\in \mathcal{M}_n, k \in \mathcal{K}, \label{eq:singlesecond} \displaybreak[1]\\
	&  x_{ik}^{(n)} \geq 0, \quad \forall\; i\in \mathcal{M}_n, k \in \mathcal{K}. \label{eq:singlethird} \displaybreak[2]
	\end{align}
\end{subequations}
The above optimization problem captures the goal of every single edge cloud i.e., maximizing its utility by allocating its available resources to its native applications\footnote{Applications that are primarily affiliated with the edge cloud.}. The first constraint \eqref{eq:singlefirst} indicates that  resources allocated to all users $i \in \mathcal{M}_n$ cannot exceed capacity. The second constraint \eqref{eq:singlesecond} indicates  allocated resources cannot exceed  required resources while the final constraint \eqref{eq:singlethird} requires the allocations to be non-negative.  For the cooperative case where edge clouds work in cooperation with other edge clouds, we aim to maximize the utility of our coalition in~\eqref{eq:opt_higher}:  
\begin{subequations}\label{eq:opt_higher}
	\begin{align}
\max_{\mathcal{X}}\quad&  \big(w_n u_n(\mathcal{X})+ \zeta_{n}\sum_{j \in \mathcal{N},j\neq n}  u^{n}_j(\mathcal{X})\big) \quad \forall n \in \mathcal{N} \label{eq:obj}\\
	\text{s.t.}\quad & \sum_{i \in \mathcal{M}} x_{ik}^{(n)}\leq C_{k}^{(n)}, \quad \forall k \in \mathcal{K}, \quad \forall n \in \mathcal{N}, 
	 \label{eq:obj1} \displaybreak[0]\\
	& \sum_{j \in \mathcal{N}}x_{ik}^{(j)} \leq r^{(n)}_{ik}, \quad \forall\; i\in \mathcal{M}, k \in \mathcal{K}, n \in \mathcal{N},  \label{eq:obj2}\displaybreak[1]\\
	&  x_{ik}^{(n)} \geq 0, \quad \forall\; i\in \mathcal{M}, k \in \mathcal{K}, n \in \mathcal{N}. \label{eq:obj3} \displaybreak[2]
	\end{align}
\end{subequations}
 Here $u_n(\mathcal{X})$ in ~\eqref{eq:opt_higher} indicates the utility that an edge cloud receives by providing its resources to  applications in  $\mathcal{M}_n$. Remaining available resources at edge cloud $n$ 	can be used by applications of other edge clouds $j \in \mathcal{N}\backslash n$ that will be charged at a rate that edge cloud (each edge cloud acts as a player in our cooperative game) $j$ would have charged the application $i \in \mathcal{M}_j$ had the request been satisfied by edge cloud $j$. Hence $u^{n}_j(\mathcal{X})$ is the utility that edge cloud $n$ receives for sharing its resources with edge cloud $j$. 
 $w_n$ is the weight assigned to the utility of player $n$. $\zeta_n$ is the weight assigned to the utility $u^{n}_j(\mathcal{X})$. The purpose of the weights is to highlight that each edge cloud first allocates resources to its own applications and then shares the remaining resources with other edge clouds. 


%% file: game-theory.tex
The characteristic function $v$ for our game that solves  problem in~\eqref{eq:opt_higher} is given in ~\eqref{eq:payoff_function}. We model the resource allocation and sharing problem (with multiple objectives) in the aforementioned settings  as a canonical cooperative game with transferable utility. We choose a \emph{monotone non-decreasing utility function} for our resource allocation and sharing framework. This is because in edge computing,  the more resources provided, the higher is the gain or utility for the edge cloud. It is highly unlikely that the utility of any edge cloud will decrease with an increase in the amount of resources allocated.   Since the utility function used is monotone non-decreasing, the game is convex. The core of any convex game is large and contains the Shapley value \cite{han2012game}. Our goal is to obtain an allocation from the core as all  allocations in the core guarantee Pareto optimality and stability of the grand coalition i.e. the allocation decision obtained is Pareto optimal and no player (edge cloud) will have the incentive to leave the grand coalition. 
We first use the Shapley value, that requires solving  $2^{N}-1$ optimization problems,   to obtain an allocation from the core and then propose a computationally efficient algorithm that can provide us an allocation from the core but does not provide the fairness of the Shapley value. 
\begin{align} \label{eq:payoff_function}
v(\mathcal{N})=
\sum_{n \in \mathcal{N}}\bigg(w_n u_n(\mathcal{X})+ \zeta_{n}\sum_{j \in \mathcal{N},j\neq n} u^{n}_j(\mathcal{X})\bigg)
\end{align}

Algorithm \ref{algo:alg1} provides an overview of our proposed approach. 
We calculate the Shapley value for the players and assign it to $\mathbf{u}(R,\mathcal{X})$. Finally to obtain $\mathcal{X}$, we take the inverse function of $\mathbf{u}$. As we are using monotonic utilities, we know that the inverses of the utilities exist.  A fundamental issue with the Shapley value is its complexity. 
This motivates developing a more efficient algorithm to obtain an allocation from the core.

\subsection{Reducing the Computational Complexity}
To reduce computational complexity, we propose an algorithm  (Algorithm \ref{algo:alg2}) that requires solving only $2N$ optimization problems rather than $2^{N}$.
\begin{algorithm}
	\begin{algorithmic}[]
		\State \textbf{Input}: $R, C,$ and vector of utility functions of all players $\mathbf{u}$ 
		\State \textbf{Output}: $\mathcal{X}$
		\State \textbf{Step $1$:}  $ \mathbf{u}(R,\mathcal{X}) \leftarrow$0,  $ \mathcal{X} \leftarrow$0, $ \boldsymbol{\phi} \leftarrow$0,  
		\State \textbf{Step $1$:}  Calculate Shapley Value $\phi_n\;$  $\forall n \in \mathcal{N}$
		\State \textbf{Step $2$:} $ \mathbf{u}(R,\mathcal{X}) \leftarrow \boldsymbol{\phi}$
		\State \textbf{Step $3$:} $\mathcal{X}\leftarrow \mathbf{u}^{-1}$
	\end{algorithmic}
	\caption{Shapley Value based Resource Allocation}
	\label{algo:alg1}
\end{algorithm}
\begin{algorithm}
	\begin{algorithmic}[]
		\State \textbf{Input}: $R, C,$ and vector of utility functions of all players $\mathbf{u}$ 
		\State \textbf{Output}: $\mathcal{X}$
		\State \textbf{Step $1$:}  $ \mathbf{u}(R,\mathcal{X}) \leftarrow$0,  $ \mathcal{X} \leftarrow$0,   $ \boldsymbol{O_1} \leftarrow$0, $ \boldsymbol{O_2} \leftarrow$0, 
		\State \textbf{Step $2$:}
		 \For{\texttt{$n \in \mathcal{N}$}}
		\State $ {O_1^n} \leftarrow$ \texttt{Solution of the optimization problem in Equation \eqref{eq:opt_single}}
		\EndFor
		\State \textbf{Step $3$:} Update $R, C$ based on Step 2
		\State \textbf{Step $4$:}
		\For{\texttt{$n \in \mathcal{N}$}}
		\State $ {O_2^n} \leftarrow$ \texttt{Solution of the optimization problem in Equation \eqref{eq:opt_single_j} with updated $R$ and $C$}
		\State Update $R$ and $C$
		\EndFor
		\State \textbf{Step $5$:} $u_n(R,\mathcal{X})\leftarrow O_1^n+O_2^n$ $\forall n \in \mathcal{N}$ 
		\State \textbf{Step $6$:} $\mathcal{X}\leftarrow \mathbf{u}^{-1}$
	\end{algorithmic}
	\caption{$\mathcal{O}$($N$) algorithm for obtaining Core's allocation}
	\label{algo:alg2}
\end{algorithm}
\begin{subequations}\label{eq:opt_single_j}
	\begin{align}
	\max_{\mathcal{X}}\quad& \sum_{j\neq n}u_j^n(\mathcal{X}) \quad  \forall n \in \mathcal{N} \label{eq:objsinglej}\\
	\text{s.t.}\quad & \sum_{i \in \mathcal{M}\backslash\mathcal{M}_n} x_{ik}^{(n)}\leq C_{k}^{(n)}, \quad \forall k \in \mathcal{K}, 
	 \label{eq:singlefirstj} \displaybreak[0]\\
	& x_{ik}^{(n)} \leq r^{(n)}_{ik}, \quad \forall\; i\in \mathcal{M}\backslash\mathcal{M}_n, k \in \mathcal{K}, \label{eq:singlesecondj} \displaybreak[1]\\
	&  x_{ik}^{(n)} \geq 0, \quad \forall\; i\in \mathcal{M}\backslash\mathcal{M}_n, k \in \mathcal{K}. \label{eq:singlethirdj} \displaybreak[2]
	\end{align}
\end{subequations}
\begin{theorem}
	The solution obtained from Algorithm \ref{algo:alg2} lies in the \emph{core}.
\end{theorem}
\begin{proof}
	We need to show that the utilities obtained in Step $5$ of Algorithm \ref{algo:alg2}: a) are individually rational. b) are group rational, and c) no player  has the incentive to leave the grand coalition and form another coalition $S \subset \mathcal{N}$. 
	
	\noindent\textit{Individual Rationality:} For each player $n \in \mathcal{N}$, the solution $O_1^n$ obtained by solving the optimization problem in \eqref{eq:opt_single} is the utility a player obtains by working alone i.e. it is $v\{n\}$. The value $O_2^n$ in Step 4 is either zero or positive but cannot be negative due to the nature of utility used i.e. 
	\begin{align}
	u_n(R,\mathcal{X})&=O_1^n+O_2^n \geq O_1^n, \quad \forall n \in \mathcal{N}. \nonumber
	\end{align}
\noindent	\textit{Group Rationality:} The value of the grand coalition $v\{\mathcal{N}\}$ as per  \eqref{eq:payoff_function} is the sum of utilities $u_n$'s and $u_j^n$'s. Steps 2, 4 and 5 of Algorithm \ref{algo:alg2} obtain the sum of the utilities as well. Hence the solution obtained as a result of Algorithm \ref{algo:alg2} is group rational.    Furthermore, due to super-additivity of the game and monotone non-decreasing nature of the utilities, no player has the incentive to form a smaller coalition. Hence Algorithm \ref{algo:alg2} provides a solution from the core.  
\end{proof}

%% file: 04-exp-results.tex
To evaluate the performance of the proposed resource sharing and allocation mechanism, we consider a set of game players where each player has three different types of resources i.e., storage, communication and computation. Without loss of generality (W.l.o.g), the model can be extended to include other types of resources/parameters as well. We present results for four different settings with 1) $N=3,\; {M}_n=3, \; \forall n \in \mathcal{N}$; 2) $N=3,\; {M}_n=20, \; \forall n \in \mathcal{N}$;   3) $N=3,\; {M}_n=100, \; \forall n \in \mathcal{N}$; and 4) $N=10,\; {M}_n=20, \; \forall n \in \mathcal{N}$. 
We used linear and sigmoidal utilities (see ~\eqref{eq:ut1}) for all the players. However, the results  hold for any monotone non-decreasing utility.
\begin{align}\label{eq:ut1}
u_n(\mathcal{X})=\sum_{i \in \mathcal{M}_n}\bigg(\sum_{k=1}^{{K}}\frac{1}{1+e^{-\mu(x_{ik}^{(n)}-R^{(n)}_{ik})}}\bigg) \quad \forall n \in \mathcal{N}. 
\end{align}
$\mu$ is chosen to be either $0.01$  or $10$ to capture the requirements of different applications.  
The request matrix $R^{(n)}$ and capacity vector $C^{(n)}$ for each player $n \in \mathcal{N}$ were randomly generated. 
The experiments were run in \texttt{Matlab R2016b} on \texttt{Core-i7} processor with \texttt{16  GB RAM}. The optimization problems were solved using the \texttt{OPTI-toolbox}.

\subsection{Value of Coalition}
 Table \ref{tab:util1} shows the utilities of different players in a $3$ player-$3$ application setting. Player one was assigned a linear utility while player two and three had sigmoidal utilities. It is evident from the table  that the utility of the coalition improves when more players are added. The grand coalition has the highest utility, verifying the superadditive nature of the game. The last row shows the Shapley Values (S.V.) for the grand coalition. Our $\mathcal{O}(N)$ (alg2 in Figure \ref{fig:3playerm01}) provides the same value of coalition as Shapley value (due to Pareto optimality), however players are assigned different utilities in the coalition. 
 Figure \ref{fig:3playerm01} shows the value of coalition for $3$ players, and $3,10$ and $100$ applications with $\mu$ set to $0.01$ and $100$. The grand coalition achieves the highest coalition utility for all  three cases. As a higher value of $\mu$ (i.e., the slope of the sigmoidal function is steep) puts a stringent requirement on the edge clouds to satisfy requests of the applications if it is to gain any utility, we see that that the overall value of coalition is smaller for $\mu=10$ when compared with $\mu=0.01$. Figure \ref{fig:all} shows the utility of a player without resource sharing and with resource sharing in the grand coalition in varying experimental settings ($\mu=0.01$, $\mu=10$, $M_n=20$, and $M_n=100$). We show the utility of the player in the grand coalition both using Shapley value (SV) and our $\mathcal{O}(N)$ algorithm (alg2 in Figure). Similar trends are observed in Figure \ref{fig:all2}. However, we do not calculate the Shapley value for $N=10$ and $M_n=20$ due to the computational complexity and the utility of the players in the grand coalition is obtained using Algorithm $2$.   It is evident that all the players' utilities improve by sharing resources and taking part in the cooperative game.  

\begin{table}[]
	\centering
	\caption{Utility (Pay-off) for different coalitions with $\mu=0.01, N=3, M_n=3$}
	\begin{tabular}{|l|l|l|l|l|}
		\hline
		\multirow{2}{*}{\textbf{Coalition}} & \multicolumn{3}{c|}{\textbf{Player Utilities}}                     & \multirow{2}{*}{\textbf{Coalition Utility}} \\ \cline{2-4}
		& \textbf{$P_1$} & \textbf{$P_2$} & \textbf{$P_3$} &                                             \\ \hline
		\{1\}	&      36            &        0           &      0             &    36                                    \\ \hline
		\{2\}	&      0             &     4.37             &       0            &    4.37                                 \\ \hline
		\{3\}	&         0          &    0               &        4.31          &   4.31                   \\ \hline
		\{12\}	&        40.17          &     4.375              &    0               &         44.545     \\ \hline
		\{13\}	&       40.31          &      0             &     4.313              &  44.623                     \\ \hline
		\{23\}	&        0           &      8.68            &   8.68                &    17.37                \\ \hline
		\{123\} $\mathcal{O}(N)$	&   44.68                &   8.68               & 8.68                &  62.06      \\ \hline
		\{123\} (S.V.)	&   40.34              &    10.90               &   10.81            &  62.06      \\ \hline
	\end{tabular}
	\label{tab:util1}
\end{table}
	
 \begin{figure}
	\includegraphics[width=0.55\textwidth]{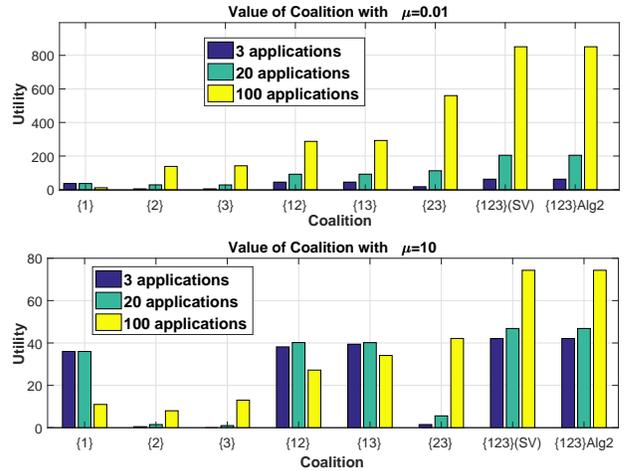}
	\caption{Value of Coalition for $3$ players, and $3,20$ and $100$ applications with $\mu=0.01$ and $\mu=10$}
	\protect\label{fig:3playerm01}
\end{figure}

\subsection{Time complexity} 
Computational complexity of the Shapley value is high, which is why  it cannot be used for a large number of players. We compared the performance of our $\mathcal{O}(N)$ algorithm with the Shapley value based allocation (given in Algorithm \ref{algo:alg1}) in a 3-player game with different number of applications. Experimental results showed that Algorithm \ref{algo:alg2} reduces the computation time by as large as 71.67\% and as small as 26.6\% while the average improvement was about 49.75\%. Figure \ref{fig:time} shows the calculation time for different user-application settings with varying $\mu$. We see that in all the settings, our proposed algorithm outperforms the Shapley value based allocation. 

 \begin{figure}
	\includegraphics[width=0.55\textwidth]{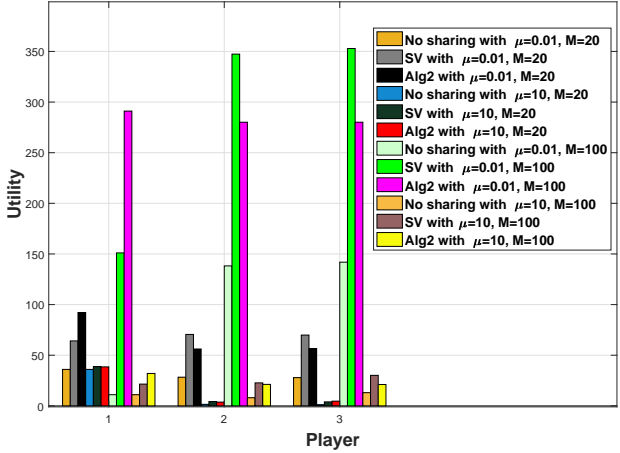}
	\caption{Utilities in different settings without and with resource sharing in grand coalition for $N=3$ }
	\protect\label{fig:all}
\end{figure}

\begin{figure}
	\includegraphics[width=0.55\textwidth]{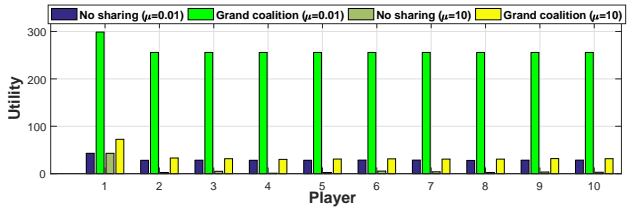}
	\caption{Player utilities  with and without resource sharing in grand coalition for $N=10$ and $M=20$ }
	\protect\label{fig:all2}
\end{figure}

 \begin{figure}
	\includegraphics[width=0.55\textwidth]{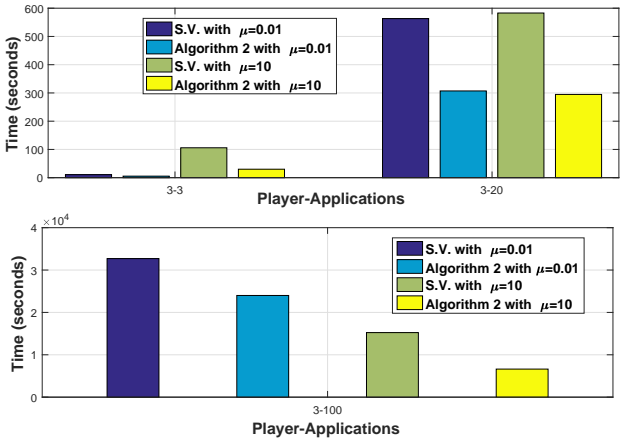}
	\caption{Comparison of time Complexity}
	\protect\label{fig:time}
\end{figure}

%% file: 05-conclusions.tex
We proposed a cooperative game theory based  resource allocation and sharing framework for edge computing that can efficiently allocate resources to different applications affiliated with edge clouds. Our resource allocation and sharing game is canonical and convex. The core for the game is non-empty, hence the grand coalition is stable and Shapley value also lies in the core. Furthermore, due to computational complexity of calculating Shapley value, we presented an $\mathcal{O}(N)$ algorithm that can provide an allocation and sharing decision from the core. Experimental results showed that edge clouds can improve their utility by using our proposed resource allocation mechanism and our  $\mathcal{O}(N)$ algorithm can provide us an allocation from the core (guarantee of Pareto optimality and stability) in a shorter time than the Shapley value.